\documentclass[conference]{IEEEtran}
\IEEEoverridecommandlockouts

\usepackage{cite}
\usepackage{amsmath,amssymb,amsfonts}
\usepackage{amsthm}
\usepackage{algorithmic}
\usepackage{graphicx}
\usepackage{textcomp}
\usepackage{xcolor}
\usepackage{booktabs}
\usepackage[colorlinks=true,allcolors=blue]{hyperref}
\hypersetup{pdftitle={Quantifying the Relay-Coverage Trade-off in Mobile Multihop Flooding},
  pdfauthor={Johannes F. Loevenich, Thies M\"ohlenhof, Roberto Rigolin F. Lopes}}
\usepackage{balance}  

\usepackage[acronym,nomain]{glossaries}
\glsdisablehyper
\setacronymstyle{long-short}
\glsnoexpandfields
\newcommand*{\olsrvtwolong}{\ifglsused{olsr}{\glsentryshort{olsr}}{\glsentrylong{olsr}} version~2}
\newacronym{olsr}{OLSR}{Optimized Link State Routing}
\newacronym{olsrv2}{OLSRv2}{\olsrvtwolong}
\newacronym{nhdp}{NHDP}{Neighborhood Discovery Protocol}
\newacronym{smf}{SMF}{Simplified Multicast Forwarding}
\newacronym{ecds}{E-CDS}{Essential Connected Dominating Set}
\newacronym{mpr}{MPR}{multipoint relay}
\newacronym{tc}{TC}{Topology Control}
\newacronym{rfc}{RFC}{Request for Comments}
\newacronym{ip}{IP}{Internet Protocol}
\newacronym{ci}{CI}{confidence interval}
\newacronym{cpsat}{CP-SAT}{constraint programming--satisfiability}

\theoremstyle{plain}

\newtheorem{proposition}{Proposition}
\newtheorem{corollary}{Corollary}

\theoremstyle{definition}
\newtheorem{definition}{Definition}

\newcommand{\Nl}{N_1}
\newcommand{\Nt}{N_2}

\newcommand{\rcov}{r}
\newcommand{\hold}{L}

\begin{document}

\title{
    Quantifying the Relay-Coverage Trade-off in Mobile Multihop Flooding
}

\author{\IEEEauthorblockN{Johannes F. Loevenich\IEEEauthorrefmark{1}\IEEEauthorrefmark{2},
Thies M\"ohlenhof\IEEEauthorrefmark{1},
Roberto Rigolin F. Lopes\IEEEauthorrefmark{1}}
\IEEEauthorblockA{\IEEEauthorrefmark{1}Thales Deutschland, cortAIx Labs, Germany\qquad \\
\IEEEauthorrefmark{2}Institute of Computer Science, University of Osnabr\"uck, Germany\\
\{johannes.loevenich, thies.moehlenhof, roberto.rigolin\}@thalesgroup.com}}

\maketitle

\begin{abstract}
Network-wide flooding in mobile multihop networks selects relays from locally maintained topology information. Mobility makes this information outdated, leaving two basic responses: add redundant relays or refresh the topology and re-select relays more frequently.
%
How these mechanisms trade reliability for overhead, however, is not well quantified.
%
We quantify this trade-off on a public trace of 157 vehicles using \gls{olsrv2} and exact minimum-cost relay sets.
%
Twofold coverage increases median relay-cover lifetime by only $2.3\times$, rather than about an order of magnitude as predicted by an independent-hazard model.
%
The limitation is structural: $43.6\,\%$ of two-hop neighbours have only one possible relay, and newly appearing neighbours invalidate existing covers. Moreover, the delivery gain of a second relay falls from $0.061$ to $0.016$ when relay selection uses instantaneous link state, showing that redundancy mainly compensates for outdated topology information.
%
Relay coverage and topology refresh therefore form a measurable redundancy--freshness trade-off rather than independent robustness mechanisms.
%
In our trace, twofold coverage has lower overhead below about $100\,\mathrm{bit/s}$ of application flooding per node, whereas fresher topology information has lower overhead above it. This makes relay redundancy a load-dependent design choice rather than a fixed heuristic robustness setting.
\end{abstract}

\glsresetall

\section{Introduction}
\label{sec:intro}

Network-wide broadcast supports route discovery and control dissemination in infrastructureless wireless networks, but naive flooding wastes capacity through redundant transmissions and collisions~\cite{ni1999broadcaststorm}. Relay-based schemes instead select one-hop neighbours that collectively reach the two-hop neighbourhood, as in \gls{olsr}/\gls{olsrv2}~\cite{qayyum2002mpr,clausen2003olsr,clausen2014olsrv2}, \gls{smf}~\cite{macker2012smf}, and dominating-set broadcast~\cite{dai2006kconnected}. Mobility creates a common failure mode: the topology changes while the local information used to choose relays ages.

A protocol can tolerate this mismatch in two ways. It can add relay redundancy, spending forwarding capacity so that outdated choices remain useful, or maintain a fresher topology view, spending control capacity so obsolete choices are replaced sooner. \gls{olsr} exposes the former through \texttt{MPR\_COVERAGE} and explicitly suggests increasing coverage under mobility~\cite[\S16]{clausen2003olsr}. Yet the exchange between these mechanisms, how much robustness a second relay buys, why, and when that cost is preferable to fresher topology information, has not been quantified on a realistic mobility trace (Section~\ref{sec:related}).

The answer need not follow independent-backup intuition. Some two-hop neighbours have only one possible relay, new neighbours can invalidate an otherwise intact cover, and a router can continue to believe in a link after it breaks. Redundancy may therefore be limited by local topology and may act less as protection against instantaneous relay loss than as insurance against outdated information.

We measure this exchange on a public trace of 157 vehicles, using \gls{olsrv2}/\gls{nhdp} as the primary instance because its neighbourhood state and forwarding semantics are explicit. An exact minimum-cost selector provides a yardstick against heuristic selector quality; \gls{smf} and collision arms test whether the mechanism survives a different relay rule and intra-flood contention.

The measurements show that twofold coverage extends median cover lifetime by only $2.3\times$, not by an order of magnitude. More importantly, most of its delivery gain disappears when relay selection receives current link state. Redundancy and topology freshness are therefore alternative ways of spending overhead against the same topology uncertainty, and application flooding load determines which is cheaper.

This paper makes three contributions:
\begin{itemize}
\item \textbf{Trade-off.} We quantify the reliability--overhead exchange between relay redundancy and topology freshness under mobility.
\item \textbf{Mechanism.} We show that structural coverability limits redundancy and that its delivery gain mainly compensates for outdated topology belief.
\item \textbf{Design rule.} We identify the application-load crossover at which fresher topology becomes cheaper than twofold relay coverage; the underlying mechanism also appears with \gls{smf}.
\end{itemize}
Together, these results turn relay redundancy from a fixed robustness heuristic into a measurable, load-dependent design choice for mobile multihop flooding.

\section{Related Work}
\label{sec:related}

\emph{Relay flooding.} Relay-based broadcast reduces the broadcast storm by restricting forwarding to a backbone. Multipoint relaying lets each node select forwarders from its one-hop neighbourhood~\cite{qayyum2002mpr,clausen2014olsrv2}; \gls{smf}'s \gls{ecds} instead lets nodes self-elect from their two-hop views~\cite{macker2012smf}; connected dominating sets and gossiping provide structural or probabilistic alternatives~\cite{dai2006kconnected,haas2006gossip}. These schemes differ in how relays are chosen, but all face the same mobility problem: forwarding decisions are based on a local view that can age before the next flood.

\emph{Redundant coverage.} \gls{olsr}'s \texttt{MPR\_COVERAGE}$=m$ asks that every two-hop neighbour be reached through $m$ \glspl{mpr} \textit{``if possible''} and gives poorly covered nodes all available coverers~\cite{clausen2003olsr}; \gls{olsrv2} no longer exposes this parameter. Double coverage has been studied as protection against lossy links~\cite{lou2007double}, while range buffers protect against outdated neighbour views~\cite{wu2005guaranteed}. Nilsson and Sterner compare additional \glspl{mpr} with fading margins or hysteresis and find that the preferred robustness mechanism depends on topology density~\cite{nilsson2012robust,nilsson2015rural}. This establishes that redundancy can help, but not how much temporal slack a second relay actually buys or when its recurring forwarding cost is preferable to a fresher view.

\emph{Stable selection and topology freshness.} A complementary line tries to make a chosen backbone last longer by preferring stable links, predicting link lifetime, using mobility models, or retaining \glspl{mpr} across runs~\cite{obilisetty2005lsolsr,ouacha2013reliable,li2017smlr,maccari2012reduce}. These approaches choose better relays from the information currently available. Separately, HELLO intervals and hold times have been tuned to make that information more current~\cite{huang2006tuning,ernst2012adaptive,jovel2017holdtimes}. The two lines therefore act on different sides of the same mismatch: one tolerates or avoids topology change, the other reduces uncertainty about it.

\emph{Gap.} Prior work studies redundant coverage, longer-lived relay selection, and neighbourhood refresh largely as separate mechanisms. We instead measure their common exchange: how much outdated-topology error can be tolerated by additional relay coverage, how structural coverability limits that protection, and at what application load refreshing the topology becomes cheaper than forwarding through a redundant backbone.

\section{System Model}
\label{sec:model}

The key distinction is between the true mobile topology and the local topology belief on which relay decisions are made. At time $t$, let $G_t=(V,E_t)$ be the true graph, and let $\Nl^t(v)$ and $\Nt^t(v)$ denote node $v$'s one-hop and strict two-hop neighbours; members of $\Nt^t(v)$ are \emph{terminals}. Node $v$ selects relays $R_t(v)\subseteq\Nl^t(v)$.

\begin{definition}[$\rcov$-cover]
\label{def:rcover}
$R_t(v)$ is an \emph{$\rcov$-cover} if every terminal $u\in\Nt^t(v)$ is adjacent to at least $\rcov$ relays in $R_t(v)$. Standard coverage has $\rcov=1$.
\end{definition}

Let $\widehat G_t(v)$ be $v$'s believed two-hop topology. Under perfect knowledge it equals the true local topology; under \gls{nhdp}~\cite{clausen2011nhdp}, periodic HELLOs update it and advertised entries persist for a hold time $T_{\mathrm{hold}}$. Relay sets are selected on $\widehat G_t(v)$ but used on $G_t$. Under \gls{olsrv2} forwarding~\cite{clausen2014olsrv2}, a relay can act for a selector only after that selection has been advertised, so outdated links, outdated membership, and relay-flag lag can all affect delivery.

The two control levers spend different resources. Increasing $\rcov$ can increase the number of nodes that forward each application flood; decreasing the HELLO interval $H$ or $T_{\mathrm{hold}}$ increases control traffic but makes $\widehat G_t(v)$ track $G_t$ more closely. We quantify the reliability–overhead exchange between redundancy and topology freshness and identify when each mechanism is the lower-cost choice.

\section{The Redundancy--Lifetime Exchange}
\label{sec:theory}

Before introducing topology-belief error, we derive the lifetime benefit expected if covering links fail independently and redundant cover is uniformly available. This gives a reference model for the measurements.

\begin{proposition}[Coverage lifetime]
\label{thm:lifetime}
Let covering links fail independently with hazard $\lambda$, and let a relay set reach $m_a$ of the $m$ terminals exactly $a$ times. To first order in $\lambda\hold\ll1$, the time $\hold$ by which some terminal is uncovered with probability $\varepsilon$ satisfies
\[
\sum_a m_a(\lambda\hold)^a\approx\varepsilon.
\]
For a uniform $\rcov$-cover relative to a 1-cover,
\[
\frac{\hold_\rcov}{\hold_1}=\left(\frac{m}{\varepsilon}\right)^{1-1/\rcov}.
\]
\end{proposition}
\begin{proof}
Fix the relay set at time $0$; a covering link of terminal $y$ joins $y$ to one of its $a_y$ relays. Each covering link has failed by time $\hold$ with probability $p=1-e^{-\lambda\hold}$, independently of all others. Terminal $y$ is uncovered by time $\hold$ if and only if all its covering links have failed, an event $A_y$ of probability $p^{a_y}$. Distinct terminals share no covering link, so the $A_y$ are independent and some terminal is uncovered by time $\hold$ with probability ${F(\hold)=1-\prod_y(1-p^{a_y})}$. Let $S=\sum_y p^{a_y}={\sum_a m_a p^a}$. The union bound gives $F\le S$; the Bonferroni inequality and independence give ${F\ge S-\sum_{y<y'}p^{a_y}p^{a_{y'}}}\ge S-S^2/2$. As ${x-x^2/2}\le{1-e^{-x}}\le x$ for $x\ge0$, ${\lambda\hold(1-\lambda\hold/2)}\le p\le\lambda\hold$. Let $\hold$ solve $\sum_a m_a(\lambda\hold)^a=\varepsilon$ and $\bar a=\max\{a:m_a>0\}$. Then $\varepsilon(1-\lambda\hold/2)^{\bar a}\le S\le\varepsilon$, and since $S-S^2/2$ increases on $[0,1]$,
\[
\varepsilon\Bigl(1-\frac{\lambda\hold}{2}\Bigr)^{\bar a}\Bigl(1-\frac{\varepsilon}{2}\Bigr)\le F(\hold)\le\varepsilon .
\]
Hence $F(\hold)=\varepsilon\bigl(1-O(\bar a\lambda\hold+\varepsilon)\bigr)$: to first order, $\hold$ is the time by which some terminal is uncovered with probability $\varepsilon$, and it never overstates that time. The upper bound needs no independence across terminals. For a uniform $\rcov$-cover, $m_\rcov=m$ gives $\lambda\hold_\rcov=(\varepsilon/m)^{1/\rcov}$; a 1-cover reaching every terminal once gives $\lambda\hold_1=\varepsilon/m$. Dividing, $\hold_\rcov/\hold_1=(\varepsilon/m)^{1/\rcov-1}=(m/\varepsilon)^{1-1/\rcov}$.
\end{proof}

Uniform redundancy need not exist. For terminal $y$, define its number of possible coverers
\[
\kappa(y)=|\{w\in\Nl(v):y\in\Nl(w)\}|,
\]
and cap achievable coverage at $\rcov_y=\min(\rcov,\kappa(y))$. A terminal with $\kappa(y)=1$ is \emph{pinned}: no relay budget can give it redundancy.

\begin{corollary}[Structural cap]
\label{cor:exchange-real}
For a capped $\rcov$-cover with $\rcov\ge2$, let $m_\pi>0$ be the number of pinned terminals. The first-order lifetime gain over a 1-cover is bounded by $m/m_\pi$; additional coverage elsewhere cannot remove this cap while pinned terminals remain.
\end{corollary}

The reference law is deliberately optimistic: if redundancy were uniformly available and covering-link losses independent, a second relay should buy roughly an order of magnitude of lifetime at the neighbourhood sizes observed here. The capped model predicts why this need not happen: a single pinned terminal restores a first-order failure mode that extra relays elsewhere cannot remove. The size of the gap between these two cases is therefore itself diagnostic. E1 tests whether mobility behaves closer to the uniform backup model or to the structural cap; E2 then asks whether the remaining delivery value of redundancy comes from tolerating an outdated topology belief rather than extending the physical lifetime of the relay set.

\section{Experimental Design}
\label{sec:setup}
\label{sec:selectors}

The experiments isolate three questions: E1 measures the intrinsic lifetime bought by redundancy on the true topology; E2 identifies what redundancy protects against under NHDP belief; E3 prices redundancy against fresher topology at matched delivery.

\paragraph{Trace and baseline}
We use the public Anglova Vignette-2 trace~\cite{suri2016realistic,suri2018anglova}: 157 vehicles over $7800$\,s at $1$\,Hz, with staging, road-march, and assembly phases over the same terrain. Links use a $120$\,dB pathloss threshold; the largest component contains $94.7\,\%$ of nodes on average (mean degree 31). Unless stated otherwise, results use the road march, where mobility is highest. The first $10$\,s are discarded. Baseline \gls{nhdp} sends HELLOs every $2$\,s at staggered phases and expires entries after $6$\,s; floods run on the true graph with \gls{rfc}~7181 forwarding~\cite{clausen2014olsrv2}.

\paragraph{Relay rules}
All selectors use the same two-hop view and capped coverage $\min(\rcov,\kappa)$. The greedy selector applies weighted-set-cover cost effectiveness~\cite{chvatal1979greedy}; the exact selector solves the minimum-cost capped cover with \gls{cpsat}. A robustness arm uses the \gls{rfc}~7181 Appendix-B ordering with the \gls{rfc}~3626 coverage extension. A node re-evaluates every $T$ seconds and otherwise holds its incumbent; at evaluation it keeps a still-valid set or re-selects a broken one. $T=1$\,s is the conformant point on the 1-Hz trace; larger $T$ deliberately measures slack. Blind flooding (all reached nodes forward) is a transmission reference.

\paragraph{E1: cover lifetime}
E1 isolates the intrinsic temporal value of redundancy from topology-belief error. Every $30$\,s, a fresh capped $\rcov$-cover is constructed on the true topology and then frozen while the topology evolves. The outcome is time until some current terminal loses all relays; departing terminals and lifetimes beyond $600$\,s are censored. Varying only $\rcov\in\{1,2,3\}$ therefore asks: \emph{if the decision were initially correct, how much longer would an additional relay keep it valid under mobility?}

\paragraph{E2: belief and delivery}
E2 reintroduces the information mismatch. Selectors run for $\rcov\in\{1,2,3\}$ and $T=1\ldots300$\,s under \gls{nhdp} belief and perfect knowledge; every node originates a flood every $10$\,s, and delivery is the fraction of the source's true component reached. Link and membership oracles selectively correct outdated belief while retaining relay-flag delay. The distinguishing prediction is simple: if redundancy mainly protects against physical relay loss, its gain should remain when the topology view is corrected; if it mainly protects against outdated information, that gain should collapse. HELLO loss, link hysteresis, and a denser $130$\,dB graph test sensitivity. \gls{smf}/\gls{ecds}~\cite[App.~A]{macker2012smf} with \gls{rfc}~5614 backup relays~\cite{ogier2009mdr} tests a self-elected relay rule. A collision arm adds up to $0.5$\,s forwarding jitter and loses a reception when another neighbour transmits within one packet airtime, with or without carrier sense; floods are simulated in isolation.

\paragraph{E3: price}
E3 treats the two mechanisms exposed by E2 as substitutes for a common reliability target. Redundancy pays additional forwarding cost on every application flood; fresher topology pays mainly additional periodic control cost. We therefore compare configurations at matched delivery and vary application flood rate to identify where the lower-overhead choice changes. Costs are transmitted bytes per node-minute: periodic HELLO and \gls{tc} traffic plus application-flood transmissions. The sensitivity analysis spans 48 accounting variants (message layout, \gls{ip} headers, routing-\gls{mpr} choice, and relay-change signalling); responsive signalling is charged but not credited with faster propagation. Selectors are deterministic; 95\,\% \glspl{ci} use moving-block bootstrap over the timestep series (300-s blocks, paired indices; 500--2000 resamples), with cover-lifetime intervals bootstrapped over anchors~\cite{kunsch1989jackknife}.

The experiments thus form a progression rather than three benchmarks: E1 establishes the effect, E2 identifies its mechanism, and E3 converts that mechanism into a network-design decision at matched reliability.

\section{Results: The Redundancy--Freshness Exchange}
\label{sec:experiments}

The results answer the three questions in order: how much time redundancy buys, what its delivery gain protects against, and when redundancy is cheaper than fresher topology.

\subsection{E1: How much time does redundancy buy?}
\label{sec:exp-lifetime}

\begin{figure*}[t]
\centering
\includegraphics[width=\linewidth]{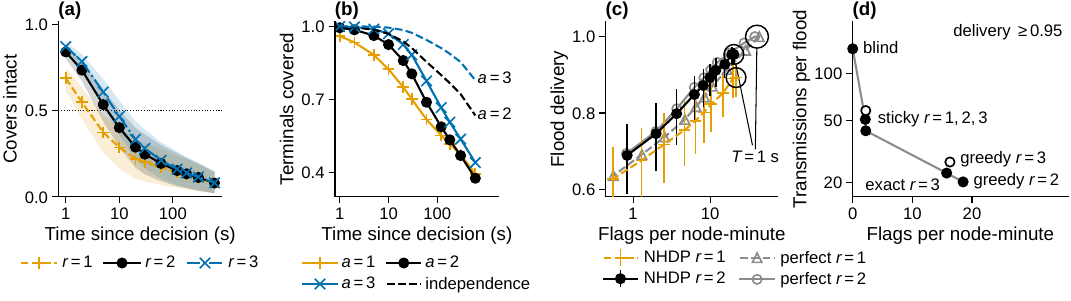}
\caption{Lifetime and delivery exchange on the road march. (a) Survival of fresh capped covers. (b) Per-terminal survival versus independent-cover prediction. (c) Delivery versus relay-flag churn under \gls{nhdp} and perfect knowledge; rings mark $T=1$\,s. (d) Configurations reaching delivery $\ge0.95$; blind means all reached nodes forward. Bars/bands: 95\,\% \glspl{ci}.}
\label{fig:exchange}
\end{figure*}

\paragraph{Independence is optimistic}
For terminals with $\kappa\ge3$, a second relay follows Proposition~\ref{thm:lifetime} within $13\,\%$ for small failure probabilities, but losses are correlated: at $60$\,s the independence prediction overstates survival by $0.17$ for both two and three relays (Fig.~\ref{fig:exchange}b).

\paragraph{Neighbourhood lifetime gains only $2.3\times$}
On the road march, $43.6\,\%$ of terminals are pinned and a uniform 2-cover exists in only $6.5\,\%$ of neighbourhoods. Median cover lifetime rises from $2.6$ to $6.0$\,s, a factor $2.3$ (95\,\% \gls{ci} $2.2$--$2.5$), versus roughly $11\times$ from the uniform reference model (Fig.~\ref{fig:exchange}a). Pinning and newly arriving terminals jointly explain much of the gap: arrivals cause $40$--$53\,\%$ of first failures, and censoring both causes raises the factor to $6.7$ (95\,\% \gls{ci} $3.7$--$18.7$). A third relay adds only $1.4\times$ over the 2-cover, unresolved by the \glspl{ci}. Thus redundancy buys time, but structural coverability and topology arrivals prevent the order-of-magnitude gain of uniform independent backup.

\paragraph{The structural cap is not just link flicker}
The coverability distribution in Fig.~\ref{fig:price}b reinforces this interpretation. The staging phase is the densest of the three mission phases yet has the largest pinned share, so simply adding neighbours does not guarantee additional relay choices. Among pinned terminal pairs that remain two hops apart one second later, $92.5\,\%$ remain pinned. Pinning is therefore mostly a persistent geometric constraint at this timescale rather than a one-snapshot threshold effect. Consistently, the measured lifetime factor falls from $3.9\times$ in neighbourhoods without a pinned terminal to $1.2\times$ where most terminals are pinned, although those strata also differ in neighbourhood size. The evidence supports the structural mechanism without claiming that pinning alone explains every failure.

\subsection{E2: What does redundancy protect against?}
\label{sec:exp-isoquant}

\paragraph{Single-cover re-selection reaches a belief-limited ceiling}
With perfect knowledge, a greedy cover re-evaluated each second delivers every flood (Fig.~\ref{fig:exchange}c). Under \gls{nhdp} belief, the conformant $T=1$\,s single cover reaches only $0.893$ (95\,\% \gls{ci} $0.845$--$0.928$); the exact single cover reaches $0.865$. A capped second cover adds $+0.061$ (paired 95\,\% \gls{ci} $0.041$--$0.088$) to greedy delivery and $+0.078$ to exact. Re-selecting faster therefore cannot remove the ceiling created by the information on which selection operates.

\begin{figure}[t]
\centering
\includegraphics[width=\columnwidth]{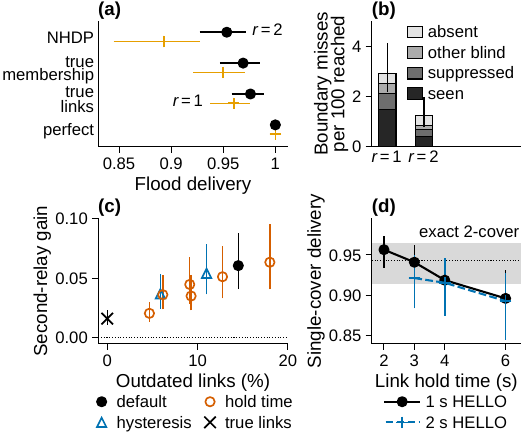}
\caption{Redundancy mostly insures against outdated topology belief (road march, greedy, $T=1$\,s). (a) Delivery as belief is corrected. (b) Boundary misses by cause. (c) Second-relay gain versus outdated covering links. (d) Single-cover delivery versus hold time; band: default exact 2-cover.}
\label{fig:belief}
\end{figure}

\paragraph{The second relay mostly insures outdated links}
On the road march, $14.5\,\%$ of believed covering links have already broken. Correcting the belief removes most of both the single-cover deficit and the redundancy gain (Fig.~\ref{fig:belief}a): with current links, single-cover delivery rises to $0.960$ and the second-relay gain falls from $+0.061$ to $+0.016$. Boundary misses show that the second relay mainly removes terminals believed covered over broken links, not terminals absent from the belief (Fig.~\ref{fig:belief}b). Across hold-time and hysteresis interventions, the gain increases with the share of outdated links (Fig.~\ref{fig:belief}c). The residual $+0.016$ may reflect other delays; relay-flag lag is one candidate. The central comparison is therefore redundancy versus topology freshness, not redundancy versus re-selection.

The oracle sequence also distinguishes outdated links from missing topology. Correcting current links removes most of the benefit, while terminals absent from the belief remain unreachable by either one or two relays because no selector can cover a terminal it does not know exists. The membership oracle behaves similarly because it also removes part of the outdated-link error. Thus the second relay is not creating new topology knowledge; it is increasing the chance that at least one path assumed by the existing belief still exists when the flood occurs.

\subsection{E3: When is redundancy cheaper than freshness?}
\label{sec:exp-price}

Reducing H\_HOLD\_TIME partially substitutes for redundancy (Fig.~\ref{fig:belief}d), but the timer cannot be shortened arbitrarily when HELLOs may be lost: \gls{rfc}~6130 recommends a hold time significantly larger than the refresh interval~\cite{clausen2011nhdp}. With 2-s HELLOs, a 3--4-s hold recovers less than half of the second relay's gain; the control-byte increase is only a few percent without HELLO loss but grows to roughly one-quarter to one-half under $10\,\%$ loss. With 1-s HELLOs and a 3-s hold, preserving the default 3:1 hold/HELLO ratio, a single cover matches the default exact 2-cover's delivery at $1.32\times$ its control bytes. Faster HELLOs alone do not help because outdated entries persist until their hold timers expire. Freshness is therefore a joint timer choice rather than simply ``send HELLOs faster.

\begin{figure}[t]
\centering
\includegraphics[width=\columnwidth]{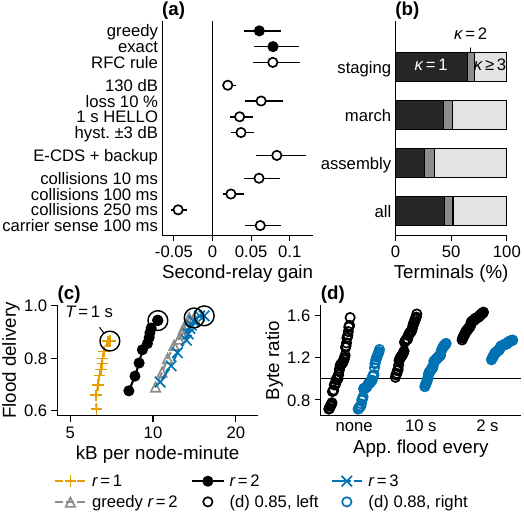}
\caption{Robustness and price. (a) Second-relay gain across selector, belief, \gls{ecds}, and collision variants. (b) Terminal coverability $\kappa$. (c) Delivery versus control bytes. (d) Twofold/single-cover byte ratio at matched delivery under 48 accounting variants and three application loads.}
\label{fig:price}
\end{figure}

\paragraph{The mechanism survives the tested variants}
The second-relay gain remains positive across all tested belief variants and shrinks as the share of outdated links falls (Fig.~\ref{fig:price}a). With $\pm3$\,dB hysteresis, outdated believed links fall from $14.5$ to $5.9\,\%$ and the gain from $+0.061$ to $+0.037$. The mechanism is not specific to \gls{olsrv2} in these experiments: \gls{ecds} delivers $0.895$ under \gls{nhdp} belief and its backup layer adds $+0.083$ (95\,\% \gls{ci} $0.057$--$0.120$), but no gain on the true graph. Under our isolated-flood collision model, the gain is essentially unchanged at $160$\,kbit/s and above; without carrier sense it falls to $+0.024$ at $100$\,ms airtime and reverses at $250$\,ms.

\paragraph{Application load sets the crossover}
At $T=1$\,s, twofold coverage roughly doubles flood transmissions: exact uses $16.2$ versus $8.3$ transmissions per flood (greedy: $20.3$ versus $10.6$). This creates different scaling laws for the two remedies. The redundant backbone pays almost nothing when no application flood is sent, but pays its extra forwarders again for every subsequent flood. A fresher topology configuration instead raises periodic control traffic whether or not the application floods, so its marginal cost per additional flood is small. Application load should therefore create a crossover even when mobility and the target delivery are fixed.

That crossover appears across the accounting variants in Fig.~\ref{fig:price}d. At the conformant comparison, the default exact 2-cover crosses a single cover with 1-s HELLOs and a 3-s hold at $3.7$ floods of 200\,B per node-minute (95\,\% \gls{ci} $2.9$--$4.8$), about $100$\,bit/s per node. Below this load, twofold coverage has lower total overhead at matched delivery because its forwarding premium is rarely paid; above it, the fixed cost of fresher topology is amortized across enough floods to become cheaper. The numerical crossover belongs to this trace and accounting, but the opposing fixed-versus-per-flood cost structure explains why a load-dependent decision arises.
The accounting sensitivity shows that the decision is not an artifact of one packet layout. With no application floods, re-selection/freshness is typically cheaper when relay changes ride on periodic HELLOs, while more expensive responsive signalling can make the held redundant backbone competitive. Once application traffic is included, the repeated forwarding cost dominates: at one 200-B flood per node every 10 s, the held 2-cover is more expensive in 89 of the 96 matched-delivery accounting cells, and at one every 2 s it is more expensive in all of them. The precise boundary moves, but the direction of the load effect is stable.

Together, E1--E3 establish the exchange: structural topology caps the lifetime benefit, outdated topology belief explains most of the delivery gain, and application load determines whether tolerating outdated belief with redundancy or correcting it with fresher updates costs less network capacity.

\paragraph{Implication for \gls{olsrv2}}
The measured exchange suggests a simple implementation rule rather than a new message format. Under light application flooding, request capped twofold coverage, i.e., cover each terminal $y$ through $\min(2,\kappa(y))$ neighbours; the existing neighbour lists already reveal $\kappa(y)$, and additional relay flags can ride in the periodic HELLO. Under heavier flooding, keep single coverage and spend the budget on a fresher neighbourhood view, represented here by the 1-s HELLO/3-s hold configuration. The rule is deliberately load-dependent: it does not claim that twofold coverage or faster refresh is universally preferable, only that the cheaper way to reach the same reliability changes with how often the forwarding backbone is exercised.

\paragraph{Transfer beyond the measured instance}
The \gls{olsrv2} timer and coverage settings are protocol-specific, but the quantities exposed by the experiments provide a more general way to analyse local-view relay schemes. A protocol designer can ask three questions in the same order: how many alternative relays each destination structurally admits, how outdated the local information is when forwarding decisions are exercised, and whether robustness is paid periodically in control traffic or repeatedly per application flood. \gls{ecds} reproduces the belief-insurance mechanism without \gls{olsrv2}'s selector/flag semantics, suggesting that the decomposition is not tied to one forwarding rule. We do not claim that another protocol shares the same crossover; rather, the result identifies the measurements needed to locate its own redundancy--freshness operating point.

\paragraph{Scope and limitations}
The collision arm models contention only among transmissions of the same flood; concurrent floods and interference from HELLO or \gls{tc} traffic are outside the model. \gls{ecds} and the belief interventions support the proposed mechanism beyond one selector, but do not establish it for every flooding protocol or mobility regime. What the experiments support more broadly is the decomposition itself: redundancy is constrained by local coverability, can mask outdated topology information, and incurs a per-flood cost that competes with the largely periodic cost of maintaining a fresher view.

\section{Conclusion}
\label{sec:conclusion}

Relays chosen on HELLO-derived belief expire in two ways: the topology moves, and the belief lags it. Re-selection repairs the first; only redundancy or a fresher belief repairs the second. On this trace, a second relay extends median cover lifetime by $2.3\times$, not by an order of magnitude: $43.6\,\%$ of two-hop neighbours admit no second relay, and newly appearing two-hop neighbours break covers too. A single cover re-selected as fast as the trace allows tops out at delivery $0.893$, while capped twofold coverage adds $+0.061$ for about twice the forwarding transmissions. With current link state, that gain falls to $+0.016$, showing that redundancy mainly insures against outdated topology belief.

Two parts of this explanation are structural: redundancy is capped wherever a two-hop neighbour has no alternative relay, and a neighbourhood belief retained beyond its observation interval can become outdated under mobility. In our experiments the same redundancy effect also appears with \gls{smf}'s \gls{ecds}, so it is not specific to \gls{olsrv2} forwarding. It also survives intra-flood collisions at $160\,\mathrm{kbit/s}$ and above under the collision model evaluated here. The numerical gains and crossover remain specific to the measured trace and protocol configuration.

\emph{For \gls{olsrv2}, the resulting rule is load-dependent.} Where an application requires more delivery than single coverage provides, capped twofold coverage and fresher topology information offer two conformant ways to obtain it. Under light application flooding, below about $3.7$ floods of 200\,B per node-minute ($\approx100\,\mathrm{bit/s}$ per node), request $\min(2,\kappa(y))$ flooding relays for each two-hop neighbour $y$. This corresponds to \acrshort{olsr}v1's \texttt{MPR\_COVERAGE}$\,=2$ semantics, requires no change to the \gls{olsrv2} message format, and incurs its additional cost primarily when floods are forwarded. Above that load, use single coverage with a 1\,s HELLO interval and 3\,s hold time instead: the additional control cost is then cheaper than paying the redundant forwarding cost on every flood. Relay coverage should therefore be selected jointly with topology freshness and application load rather than fixed as a static robustness parameter.

\bibliographystyle{IEEEtran}
\bibliography{references}

\end{document}